\newcommand*{\CTIV}{Vovk:arXiv0904}

\newcommand*{\CTXIII}{Vovk/Shafer:arXiv1703}

\documentclass[10pt]{article}
\usepackage{amsmath,amsthm,amsfonts,amssymb,latexsym,url,graphicx,euscript}
\usepackage[pdfpagemode=UseNone,pdfstartview=FitH]{hyperref}
\newcommand{\Extra}[1]{}

\allowdisplaybreaks[4]

\newcommand*{\st}{\mathrel{|}}
\newcommand*{\dd}{\,\mathrm{d}}

\renewcommand*{\S}{\mathbf{S}}
\newcommand*{\D}{\mathbf{D}}

\DeclareMathOperator{\III}{\boldsymbol{1}}
\DeclareMathOperator{\dom}{dom}

\DeclareMathOperator{\EEE}{\mathcal{E}}
\DeclareMathOperator{\LLL}{\mathcal{L}}

\newcommand*{\R}{\mathbb{R}}

\theoremstyle{plain}
\newtheorem{theorem}{Theorem}
\newtheorem{proposition}[theorem]{Proposition}
\newtheorem{corollary}[theorem]{Corollary}

\theoremstyle{definition}

  \title{Non-stochastic portfolio theory}
  \author{Vladimir Vovk}

\begin{document}
\maketitle
\begin{abstract}
  This paper studies a non-stochastic version of Fernholz's stochastic portfolio theory
  for a simple model of stock markets with continuous price paths.
  It establishes non-stochastic versions of the most basic results of stochastic portfolio theory
  and discusses connections with Stroock--Varadhan martingales.

    \bigskip

    \noindent
    The version at
    \href{http://probabilityandfinance.com/articles/index.html#51}{http://probabilityandfinance.com}
    (Working Paper 51)
    is updated most often.
\end{abstract}

\section{Introduction}

Fernholz's stochastic portfolio theory \cite{Fernholz:1999,Fernholz:2002,Fernholz/Karatzas:2009},
as its name suggests, depends on a stochastic model of stock prices.
This paper proposes a non-stochastic version of this theory based on the framework of \cite{\CTXIII}
(see the end of this section for a brief discussion of its relation to~\cite{Schied/etal:2016}).

A key finding
(see, e.g., \cite[Section~4]{Fernholz:1999}, \cite[Chapters~2 and~3]{Fernholz:2002}, \cite[Section~7]{Fernholz/Karatzas:2009})
of stochastic portfolio theory is that, under certain simplifying assumptions,
there is a long-only portfolio that outperforms the capital-weighted market portfolio.
The principal aim of this paper is to give a simple non-stochastic formalization of this phenomenon.

Section~\ref{sec:market} defines our model of a stock market
and introduces non-stochastic notions of a portfolio's value and its excess growth component.
Section~\ref{sec:master} is devoted to a non-stochastic version of the ``master equation'' of stochastic portfolio theory,
and Section~\ref{sec:special} to its applications.
In particular, the latter covers the entropy-weighted portfolio
(as in \cite[Theorem~4.1]{Fernholz:1999} and \cite[Theorem~2.3.4]{Fernholz:2002})
and diversity-weighted portfolios
(\cite[Example~3.4.4]{Fernholz:2002}, \cite[Section~7]{Fernholz/Karatzas:2009}, going back to at least \cite{Fernholz:1999ETF}).
Section~\ref{sec:interpretation} is devoted to detailed interpretations and discussions of the results of the previous sections.
Section~\ref{sec:additive} discusses connections with Stroock--Varadhan martingales,
which make the master equation very intuitive.
Finally, Section~\ref{sec:conclusion} lists some directions of further research.

Another paper treating stochastic portfolio theory in a pathwise manner is \cite{Schied/etal:2016},
and it considers a wider class of portfolios.
However, that paper relies on some assumptions that are not justified by economic considerations:
\begin{itemize}
\item
  it postulates a suitable ``refining sequence of partitions'';
\item
  it postulates the existence of a continuous covariation between each pair of price paths
  w.r.\ to this refining sequence of partitions (in F\"ollmer's \cite{Follmer:1981} sense);
\item
  a possible extension to non-smooth portfolio generating functions
  (as in \cite[Chapter~4]{Fernholz:2002}) would require postulating the existence of local times
  (perhaps along the lines of \cite{Wuermli:1980}).
\end{itemize}

\section{Market and portfolios}
\label{sec:market}

This paper uses the definitions and notation of \cite{\CTXIII} and \cite{Fernholz:2002}
(the latter, however, will always be repeated).
The notation $\int X \dd Y$ is used for the process whose value at time $t$
is $\int_0^t X(s) \dd Y(s)$, both for It\^o and Lebesgue--Stieltjes integration.
The brackets $[\ldots]$ always signify quadratic variation and are never used in the role of parentheses.
The abbreviations ``q.a.''\ and ``ucqa'' stand for ``quasi always'' and ``uniformly on compacts quasi always'';
see \cite{\CTXIII} for definitions.

We consider a financial market in which $J$ idealized securities, referred to as stocks, are traded;
their price paths $S_j:[0,\infty)\to(0,\infty)$, $j=1,\ldots,J$, are assumed to be continuous functions,
and they never pay dividends.
We let $C[0,\infty)$ stand for the set of all continuous real-valued functions on $[0,\infty)$.
As in \cite[Section~4]{\CTXIII}, we fix a sufficiently rich language for defining sequences of partitions;
all notions of non-stochastic It\^o calculus used in this paper
(such as It\^o integral and Dol\'eans exponential and logarithm)
are relative to this language.

For convenience, we identify $S_j(t)$ with the total market capitalization of the $j$th stock at time $t\in[0,\infty)$.
The \emph{total capitalization of the market} is defined as the process
\[
  S(t) := \sum_{j=1}^J S_j(t),
  \quad
  t\in[0,\infty),
\]
and the \emph{market weight} of the $j$th stock is
\[
  \mu_j(t) := S_j(t)/S(t),
  \quad
  j=1,\ldots,J.
\]
We take the total capitalization of the market as our num\'eraire,
which allows us to regard $\mu_1,\ldots,\mu_J,1$ as the traded securities
(cf.\ \cite[Section~9]{\CTXIII}),
the first $J$ of them being just like our original securities $S_j$ but constrained by $\mu_1+\cdots+\mu_J=1$.
(In fact, the original securities $S_j$ will never be used explicitly in the rest of this paper
apart from an informal remark.)

Let $\Delta^J$ be the interior of the standard simplex in $\R^J$,
\[
  \Delta^J
  :=
  \left\{
    x = (x_1,\ldots,x_J) \in (0,1)^J
    \st
    x_1 + \cdots + x_J = 1
  \right\}.
\]
A \emph{basic portfolio} is a continuous bounded function $\pi:\Delta^J\to\overline{\Delta^J}$
mapping $\Delta^J$ to its closure in $\R^J$;
intuitively, it maps the current market weights $\mu=(\mu_1,\ldots,\mu_J)$
to the fractions $\pi(\mu)=(\pi_1(\mu),\ldots,\pi_J(\mu))$ of the current capital invested in the $J$ stocks.
(In this paper we will only need these very primitive Markovian portfolios.)

The non-stochastic notions of Dol\'eans exponential $\EEE$ and Dol\'eans logarithm $\LLL$ used in this paper
are defined in \cite{\CTXIII}.
The most useful for us interpretation of Dol\'eans logarithm
is that $\LLL(Y)$ is the cumulative return of a price path $Y\in C[0,\infty)$,
and Dol\'eans exponential restores the price path from its cumulative return.
The \emph{value process} of $\pi$ is the Dol\'eans exponential
\begin{align}
  Z_{\pi}
  &:=
  \EEE
  \Bigl(
    \int\pi(\mu)\dd\LLL(\mu)
  \Bigr)
  \notag\\
  &:=
  \EEE
  \left(
    \sum_{j=1}^J
    \int\pi_j(\mu)\dd\LLL(\mu_j)
  \right)
  =
  \EEE
  \left(
    \sum_{j=1}^J
    \int
    \frac{\pi_j(\mu)}{\mu_j}
    \dd
    \mu_j
  \right),
  \label{eq:value}
\end{align}
where $\mu:[0,\infty)\to\R^J$ is defined by $\mu(t):=(\mu_1(t),\ldots,\mu_J(t))$,
$\pi_j(\mu):[0,\infty)\to\R$ is defined by $\pi_j(\mu)(t):=\pi_j(\mu(t))$,
and $\pi(\mu):[0,\infty)\to\R^J$ is defined by $\pi(\mu)(t):=(\pi_1(\mu)(t),\ldots,\pi_J(\mu)(t))$.
The value process $Z_{\pi}$ is defined and continuous quasi always.

The definition \eqref{eq:value} involves Dol\'eans logarithm,
but stochastic portfolio theory emphasizes regular logarithm
(cf.\ the logarithmic model in \cite[Section~1.1]{Fernholz:2002}).
On the log scale the definition~\eqref{eq:value} can be rewritten as
\begin{align}
  \ln
  Z_{\pi}
  &=
  \ln
  \EEE
  \left(
    \sum_{j=1}^J
    \int\pi_j(\mu)\dd\LLL(\mu_j)
  \right)\label{eq:chain-first}\\
  &=
  \sum_{j=1}^J
  \int\pi_j(\mu)\dd\LLL(\mu_j)
  -
  \frac12
  \left[
    \sum_{j=1}^J
    \int\pi_j(\mu)\dd\LLL(\mu_j)
  \right]\label{eq:chain-second}\\
  &=
  \sum_{j=1}^J
  \int\pi_j(\mu)\dd\ln \mu_j
  +
  \frac12
  \sum_{j=1}^J
  \int\pi_j(\mu)\dd[\ln\mu_j]\label{eq:third}\\
  &\qquad{}-
  \frac12
  \left[
    \sum_{j=1}^J
    \int\pi_j(\mu)\dd\ln\mu_j
  \right]
  \quad\text{q.a.}.
  \label{eq:chain-last}
\end{align}
The second equality in the chain~\eqref{eq:chain-first}--\eqref{eq:chain-last} follows from the standard equality
\begin{equation}\label{eq:standard-1}
  \EEE(X)
  =
  \exp(X-[X]/2)
  \quad
  \text{q.a.}
\end{equation}
and the third equality in~\eqref{eq:chain-first}--\eqref{eq:chain-last} follows from
\begin{equation}\label{eq:standard-2-used-in-chain}
  \LLL(Y)
  =
  \ln Y_t + \frac12 [\ln Y]
  \quad
  \text{q.a.}
\end{equation}
(showing that the first term in~\eqref{eq:chain-second} can be represented as \eqref{eq:third})
and a slight generalization of
\begin{equation}\label{eq:standard-3}
  [\LLL(Y)]=[\ln Y]
  \quad
  \text{q.a.}
\end{equation}
(showing that the second term in~\eqref{eq:chain-second} can be rewritten as \eqref{eq:chain-last}).
See \cite[Section~7]{\CTXIII} for~\eqref{eq:standard-1}--\eqref{eq:standard-3}.

The part
\begin{align}
  \Gamma^*_{\pi}
  &=
  \frac12
  \sum_{j=1}^J
  \int\pi_j(\mu)\dd[\ln\mu_j]
  -
  \frac12
  \left[
    \sum_{j=1}^J
    \int\pi_j(\mu)\dd\ln\mu_j
  \right]
  \label{eq:excess}\\
  &=
  \frac12
  \sum_{j=1}^J
  \int\pi_j(\mu)\dd[\ln\mu_j]
  -
  \frac12
  \sum_{i,j=1}^J
  \int
  \pi_i(\mu)\pi_j(\mu)
  \dd[\ln\mu_i,\ln\mu_j]
  \notag
\end{align}
of \eqref{eq:chain-first}--\eqref{eq:chain-last} consisting of the last two addends
will be called the \emph{excess growth term}
(it corresponds to the cumulative excess growth rate in stochastic portfolio theory).
We can use it to summarize~\eqref{eq:chain-first}--\eqref{eq:chain-last} as
\begin{equation}\label{eq:summary}
  \ln Z_{\pi}
  =
  \sum_{j=1}^J
  \int\pi_j(\mu)\dd\ln\mu_j
  +
  \Gamma^*_{\pi}
  \quad\text{q.a.}
\end{equation}
The addend $\sum_{j=1}^J\int\pi_j(\mu)\dd\ln\mu_j$ is the naive expression for the cumulative log growth in the value of $\pi$,
and $\Gamma^*_{\pi}$ is the adjustment required to obtain the true cumulative log growth.

A particularly important special case is that of the market portfolio, $\pi=\mu$.
To understand the intuition behind the excess growth term~\eqref{eq:excess} in this case,
we can rewrite $2\Gamma^*_{\mu}$ as
\begin{align}
  2\Gamma^*_{\mu}(t)
  &=
  \sum_{j=1}^J
  \int_0^t
  \mu_j(s)
  \dd[\ln\mu_j](s)
  -
  \left[
    \sum_{j=1}^J
    \int
    \mu_j
    \dd
    \ln\mu_j
  \right](t)
  \label{eq:Gamma-mu-1}\\
  &=
  \sum_{j=1}^J
  \int_0^t
  \mu_j(s)
  \dd[\ln\mu_j](s)
  =
  \sum_{j=1}^J
  \int_0^t
  \frac{\dd[\mu_j](s)}{\mu_j(s)}
  \label{eq:Gamma-mu-2}\\
  &\ge
  \sum_{j=1}^J
  \int_0^t
  \dd[\mu_j](s)
  =
  \sum_{j=1}^J
  [\mu_j](t),
  \notag
\end{align}
where we have used the fact that the subtrahend in \eqref{eq:Gamma-mu-1},
being the quadratic variation of a monotonic function
(remember that $\sum_j\mu_j=1$),
is zero.
We can see that $2\Gamma^*_{\mu}(t)$ is bounded below
by the total quadratic variation of the market weights.

\section{Master equation}
\label{sec:master}

Let $\S$ be a $C^2$ positive function defined on an open neighbourhood $\dom\S$ of $\Delta^J$ in $\R^J$.
For any $C^2$ function $F$ (such as $\ln\S$) defined on $\dom\S$
we let $D_j$ stand for its $j$th partial derivative,
\[
  D_j F(x)
  =
  \frac{\partial F}{\partial x_j}(x),
  \quad
  x=(x_1,\ldots,x_J)\in\dom\S,
\]
and $D_{ij}$ stand for its second partial derivative in $x_i$ and $x_j$,
\[
  D_{ij} F(x)
  =
  \frac{\partial^2 F}{\partial x_i \partial x_j}(x).
\]
The \emph{portfolio generated by $\S$} is defined by
\begin{equation}\label{eq:generated}
  \pi_j(x)
  :=
  \left(
    D_j \ln \S(x)
    +
    1
    -
    \sum_{k=1}^J
    x_k D_k \ln\S(x)
  \right)
  x_j.
\end{equation}
The main part of the expression in the parentheses is $D_j\ln\S(x)$;
the rest is simply the normalizing constant $c=c(x)$ making $(D_j\ln\S(x)+c)x_j$ a portfolio
(it is a constant in the sense of not depending on $j$).

Now we can state a non-stochastic version of the ``master equation'' of stochastic portfolio theory
(see, e.g., \cite[Theorem~3.1.5]{Fernholz:2002}).

\begin{theorem}\label{thm:master}
  The value process $Z_{\pi}$ of the portfolio $\pi$ generated by $\S$ satisfies
  \begin{equation}\label{eq:master}
    \ln Z_{\pi}(t)
    =
    \ln\frac{\S(\mu(t))}{\S(\mu(0))}
    +
    \Theta(t)
    \quad
    \text{q.a.},
  \end{equation}
  where
  \begin{equation}\label{eq:Theta}
    \Theta(t)
    :=
    \int_0^t
    \frac{-1}{2\S(\mu(s))}
    \sum_{i,j=1}^J
    D_{ij}\S(\mu(s))
    \dd[\mu_i,\mu_j](s).
  \end{equation}
\end{theorem}

\begin{proof}
  The middle equality~\eqref{eq:chain-second} in the chain~\eqref{eq:chain-first}--\eqref{eq:chain-last}
  gives for the left-hand side of~\eqref{eq:master}:
  \begin{align}
    \ln Z_{\pi}(t)
    &=
    \sum_j
    \int_0^t
    \frac{\pi_j(\mu(s))}{\mu_j(s)}
    \dd\mu_j(s)
    -
    \frac12
    \sum_{i,j}
    \int_0^t
    \frac{\pi_i(\mu(s))\pi_j(\mu(s))}{\mu_i(s)\mu_j(s)}
    \dd[\mu_i,\mu_j](s)\notag\\
    &=
    \sum_j
    \int_0^t
    \left(
      D_j \ln\S(\mu(s))
      +
      1
      -
      \sum_k
      \mu_k(s) D_k \ln\S(\mu(s))
    \right)
    \dd\mu_j(s)\notag\\
    &\quad{}-
    \frac12
    \sum_{i,j}
    \int_0^t
    \left(
      D_i \ln\S(\mu(s))
      +
      1
      -
      \sum_k
      \mu_k(s) D_k \ln\S(\mu(s))
    \right)\notag\\
    &\qquad{}\times
    \left(
      D_j \ln\S(\mu(s))
      +
      1
      -
      \sum_k
      \mu_k(s) D_k \ln\S(\mu(s))
    \right)
    \dd[\mu_i,\mu_j](s)\notag\\
    &=
    \sum_j
    \int_0^t
    \frac{D_j\S(\mu(s))}{\S(\mu(s))}
    \dd\mu_j(s)\label{eq:left-1}\\
    &\quad{}-
    \frac12
    \sum_{i,j}
    \int_0^t
    \frac{D_i\S(\mu(s))}{\S(\mu(s))}
    \frac{D_j\S(\mu(s))}{\S(\mu(s))}
    \dd[\mu_i,\mu_j](s)
    \quad\text{q.a.},\label{eq:left-2}
  \end{align}
  where the last equality follows from $\sum_k\mu_k=1$.
  Next we apply the It\^o formula to the function
  $
    \ln\S
  $
  on the right-hand side of~\eqref{eq:master};
  the It\^o formula still holds in our non-stochastic setting: cf.\ \cite[Section~6]{\CTXIII}.
  For the first addend on the right-hand side of~\eqref{eq:master} it gives us the expression
  \begin{align*}
    \ln\frac{\S(\mu(t))}{\S(\mu(0))}
    &=
    \sum_j
    \int_0^t
    \frac{D_j \S(\mu(s))}{\S(\mu(s))}
    \dd\mu_j(s)
    +
    \frac12
    \sum_{i,j}
    \int_0^t
    \frac{D_{ij}\S(\mu(s))}{\S(\mu(s))}
    \dd[\mu_i,\mu_j](s)\\
    &\quad{}-
    \frac12
    \sum_{i,j}
    \int_0^t
    \frac{D_i\S(\mu(s))D_j\S(\mu(s))}{\S(\mu(s))^2}
    \dd[\mu_i,\mu_j](s)
  \end{align*}
  equal, q.a., to \eqref{eq:left-1}--\eqref{eq:left-2} minus $\Theta$, as defined in~\eqref{eq:Theta}.
\end{proof}

\section{Special cases}
\label{sec:special}

A positive $C^2$ function $\S$ defined on an open neighbourhood of $\Delta^J$
is a \emph{measure of diversity} if it is symmetric and concave.
In this section we will discuss three examples of measures of diversity.

\subsection{Fernholz's arbitrage opportunity}

In \cite[Section 3.3]{Fernholz:2002}, Fernholz describes an arbitrage opportunity for his stochastic model of the market.
In the non-stochastic setting of this paper his portfolio ceases to be an arbitrage opportunity
but it is still interesting and suggests the possibility of beating the market
(as discussed in the next section).
Now we are interested in the measure of diversity
\begin{equation}\label{eq:quadratic-1}
  \S(x)
  :=
  1
  -
  \frac12
  \sum_{j=1}^J
  x_j^2.
\end{equation}
The components~\eqref{eq:generated} of the corresponding portfolio $\pi$ are
\begin{equation}\label{eq:arbitrage-portfolio}
  \pi_j(x)
  =
  \left(
    \frac{2-x_j}{\S(x)}
    -
    1
  \right)
  x_j.
\end{equation}
Now Theorem~\ref{thm:master} gives the following non-stochastic version of \cite[Example~3.3.3]{Fernholz:2002}.

\begin{corollary}\label{cor:arbitrage}
  The value process $Z_{\pi}$ of the portfolio~\eqref{eq:arbitrage-portfolio} satisfies
  \begin{equation}\label{eq:arbitrage}
    \ln Z_{\pi}(t)
    =
    \ln\frac{\S(\mu(t))}{\S(\mu(0))}
    +
    \sum_{j=1}^J
    \int_0^t
    \frac{\dd[\mu_j](s)}{2\S(\mu(s))}
    \quad
    \text{q.a.}
  \end{equation}
\end{corollary}

\begin{proof}
  Plugging $D_{ij}\S(x)=-\III_{i=j}$
  (where $\III_E$ stands for the indicator function of $E$)
  into \eqref{eq:Theta},
  we indeed obtain
  \[
    \Theta(t)
    =
    \int_0^t
    \frac{1}{2\S(\mu(s))}
    \sum_j
    \dd[\mu_j](s).
    \qedhere
  \]
\end{proof}

A slightly cruder but simpler version of Corollary~\ref{cor:arbitrage} is:

\begin{corollary}\label{cor:arbitrage-simplified}
  The value process $Z_{\pi}$ of the portfolio~\eqref{eq:arbitrage-portfolio} satisfies
  \begin{equation}\label{eq:arbitrage-simplified}
    \ln Z_{\pi}(t)
    \ge
    -\ln2
    +
    \frac12
    \sum_{j=1}^J
    [\mu_j](t)
    \quad
    \text{q.a.}
  \end{equation}
\end{corollary}

\begin{proof}
  It suffices to notice that $\S\in[1/2,1]$.
\end{proof}

\subsection{Entropy-weighted portfolio}

The archetypal measure of diversity \cite[Examples~3.1.2 and~3.4.3]{Fernholz:2002} is the entropy function
\[
  \S(x)
  :=
  -
  \sum_{j=1}^J
  x_j \ln x_j.
\]
Using~\eqref{eq:generated},
the components of the corresponding \emph{entropy-weighted portfolio} can be computed as
\begin{equation}\label{eq:entropy-portfolio}
  \pi_j(x)
  =
  -\frac{x_j\ln x_j}{\S(x)}.
\end{equation}
Calculating the drift term $\Theta$ in Theorem~\ref{thm:master},
we obtain the following corollary
(a non-stochastic version of \cite[Theorem~2.3.4]{Fernholz:2002}).

\begin{corollary}\label{cor:entropy}
  The value process $Z_{\pi}$ of the entropy-weighted portfolio $\pi$ satisfies
  \begin{equation}\label{eq:entropy}
    \ln Z_{\pi}(t)
    =
    \ln\frac{\S(\mu(t))}{\S(\mu(0))}
    +
    \int_0^t
    \frac{\dd\Gamma^*_{\mu}(s)}{\S(\mu(s))}
    \quad
    \text{q.a.}
  \end{equation}
\end{corollary}

\begin{proof}
  Plugging $D_{ij}\S(x)=-\III_{i=j}/x_j$ into \eqref{eq:Theta},
  we obtain
  \[
    \Theta(t)
    =
    \int_0^t
    \frac{1}{2\S(\mu(s))}
    \sum_j
    \frac{\dd[\mu_j](s)}{\mu_j(s)}.
  \]
  It remains to compare this expression with \eqref{eq:Gamma-mu-2}.
\end{proof}

\subsection{Diversity-weighted portfolios with parameter $p$}

Fix $p\in(0,1)$.
Define the \emph{measure of diversity with parameter $p\in(0,1)$} \cite[Example~3.4.4]{Fernholz:2002} as
\begin{equation*}
  \D_p(x)
  :=
  \left(
    \sum_{j=1}^J x_j^p
  \right)^{1/p}.
\end{equation*}
The \emph{$p$-diversity-weighted portfolio} has components
\begin{equation}\label{eq:D-p-portfolio}
  \pi_j(t)
  :=
  \frac{\mu_j(t)^p}{\sum_{i=1}^J \mu_i(t)^p}.
\end{equation}

The following corollary is a non-stochastic version of \cite[Example~3.4.4]{Fernholz:2002}.

\begin{corollary}\label{cor:D-p}
  The value process $Z_{\pi}$ of the diversity-weighted portfolio $\pi$ with parameter $p\in(0,1)$
  satisfies
  \begin{equation}\label{eq:D-p}
    \ln Z_{\pi}(t)
    =
    \ln\frac{\D_p(\mu(t))}{\D_p(\mu(0))}
    +
    (1-p)
    \Gamma^*_{\pi}(t)
    \quad
    \text{q.a.}
  \end{equation}
\end{corollary}

\begin{proof}
  Now \eqref{eq:Theta} gives
  \begin{align*}
    \Theta(t)
    &=
    \int_0^t
    \frac{-1}{2\D_p(\mu(s))}\\
    &\quad{}\times
    \Biggl(
      \sum_{i,j}
      (1-p)
      \left(
        \mu_i(s)
        \mu_j(s)
      \right)^{p-1}
      \left(
        \sum_k \mu_k(s)^p
      \right)^{1/p-2}
      \dd[\mu_i,\mu_j](s)\\
      &\quad{}+
      \sum_{j}
      (p-1)
      \left(
        \mu_j(s)
      \right)^{p-2}
      \left(
        \sum_k \mu_k(s)^p
      \right)^{1/p-1}
      \dd[\mu_j](s)
    \Biggr)\\
    &=
    \frac{1-p}{2}
    \int_0^t
    \left(
      \sum_{j}
      \pi_j(\mu(s))
      \frac{\dd[\mu_j](s)}{\mu_j(s)^2}
      -
      \sum_{i,j}
      \pi_j(\mu(s))
      \frac{\dd[\mu_i,\mu_j](s)}{\mu_i(s)\mu_j(s)}
    \right)\\
    &=
    (1-p)
    \Gamma^*_{\pi}(t).
    \qedhere
  \end{align*}
\end{proof}

Corollary~\ref{cor:D-p} immediately implies:
\begin{corollary}\label{cor:D-p-simplified}
  The value process $Z_{\pi}$ of the diversity-weighted portfolio $\pi$ with parameter $p\in(0,1)$
  satisfies
  \begin{equation}\label{eq:D-p-simplified}
    \ln Z_{\pi}(t)
    \ge
    (1-p)
    \Gamma^*_{\pi}(t)
    -
    \frac{1-p}{p}
    \ln J
    \quad
    \text{q.a.}
  \end{equation}
\end{corollary}

\begin{proof}
  Since $\D_p\in[1,J^{(1-p)/p}]$, we have
  \[
    \ln\frac{\D_p(\mu(t))}{\D_p(\mu(0))}
    \ge
    -\frac{1-p}{p} \ln J
  \]
  (cf.\ \cite[(7.6)]{Fernholz/Karatzas:2009});
  plugging this into \eqref{eq:D-p} gives~\eqref{eq:D-p-simplified}.
\end{proof}

\section{Beating the market}
\label{sec:interpretation}

The results of the previous section have striking implications for our idealized financial market.
The easiest to discuss is Corollary~\ref{cor:arbitrage-simplified}.
It can be interpreted, very informally, as the following Fisherian disjunction:
either the variation of each stock in the market decays,
in that the total quadratic variation $[\mu_j](\infty)$ of each of the $J$ market weights over $[0,\infty)$ is finite,
or we can beat the market in the sense that $\lim_{t\to\infty}Z_{\pi}(t)=\infty$
(cf.\ \cite{Fisher:1973}, p.~42).
Notice that the second alternative of the disjunction also takes care of the ``q.a.''\ in \eqref{eq:arbitrage-simplified}.

The portfolio~\eqref{eq:arbitrage-portfolio} is particularly tame
(or admissible, in Fernholz's \cite[Section 3.3]{Fernholz:2002} terminology):
it is long-only,
it never loses more than $50\%$ of its value relative to the market portfolio
(by \eqref{eq:arbitrage-simplified}),
and it never invests more than 3 times more than the market portfolio
in any of the stocks.

A more specific possible interpretation of Corollary~\ref{cor:arbitrage-simplified} is based on the efficient market hypothesis
in the form that was so forcefully advocated in the bestseller \cite{Malkiel:2016} by Burton G. Malkiel;
for him,
``the strongest evidence suggesting that markets are generally quite efficient is that professional investors do not beat the market.''
Even if there are ways to beat the market,
it is often believed that they should involve something unusual
rather than merely simple portfolios
such as \eqref{eq:arbitrage-portfolio}, \eqref{eq:entropy-portfolio}, or \eqref{eq:D-p-portfolio}
(widely known since at least 2002).
According to this interpretation, Corollary~\ref{cor:arbitrage-simplified} implies that in efficient markets
we expect market variation to die down eventually.

If we believe that the variation in our stock market will never die down,
we are forced to admit that Corollary~\ref{cor:arbitrage-simplified}
``opens the door to superior long-term investment returns
through disciplined active investment management'' \cite[Section~1.3]{Lo/Mackinlay:1999}.
This is the interpretation on which typical practical applications of stochastic portfolio theory are based
(see, e.g., \cite{Fernholz:1999ETF}, which, however,
is based on the stochastic versions of Corollaries~\ref{cor:D-p} and~\ref{cor:D-p-simplified}
rather than Corollary~\ref{cor:arbitrage-simplified}).

Corollary~\ref{cor:arbitrage-simplified} is a cruder version of Corollary~\ref{cor:arbitrage}
that replaces the first addend on the right-hand side of \eqref{eq:arbitrage} by its lower bound
and the denominator in the second addend by its upper bound.
Corollary~\ref{cor:arbitrage} is more precise in that it decomposes the growth in the portfolio's value
into two components:
one related to the growth in the diversity $\S(\mu)$ of the market weights
and the other related to the accumulation of the variation of the market weights.

It is standard in stochastic portfolio theory to assume
both that the market does not become concentrated, or almost concentrated, in a single stock
and that there is a minimal level of stock volatility;
precise versions of these assumptions are referred to as diversity and non-degeneracy, respectively.
We will see that the results of the previous section can be interpreted as saying that we can beat the market
unless it loses its diversity or degenerates.
Corollary~\ref{cor:arbitrage-simplified} says that, in fact,
the condition of non-degeneracy alone is sufficient;
this follows from the representation $\dd[\ln\mu_j]=\dd[\mu_j]/\mu_j^2$.
(But remember that our exposition is in terms of market weights $\mu_j$ rather than prices $S_j$,
which are usually used in stochastic portfolio theory.)

Corollary~\ref{cor:entropy} relies on both assumptions, diversity and non-degeneracy.
If the market maintains its diversity, we expect the first addend on the right-hand side of~\eqref{eq:entropy} to stay bounded below,
and if, in addition, the market does not degenerate, we expect the second addend to increase steadily.
As a result, the entropy-weighted portfolio outperforms the market.

To discuss Corollaries~\ref{cor:D-p} and~\ref{cor:D-p-simplified},
it is convenient to extend our discussion of $\Gamma^*_{\mu}$ given in Section~\ref{sec:market}
to more general $\Gamma^*_{\pi}$.
Let us now rewrite twice the excess growth term~\eqref{eq:excess},
$2\Gamma^*_{\pi}$, as
\begin{equation*}
  2\Gamma^*_{\pi}(t)
  =
  \int_0^t
  \sum_{j=1}^J
  \pi_j(\mu(s))
  \dd[\ln\mu_j](s)
  -
  \left[
    \sum_{j=1}^J
    \int
    \pi_j(\mu)
    \dd
    \ln\mu_j
  \right](t).
\end{equation*}
Define (using our fixed language) a sequence of partitions $T^1,T^2,\ldots$
that is fine for all processes used in this paper and set,
for a given partition $T^n=(T^n_k)_{k=0}^{\infty}$,
\begin{align*}
  \mu_{j,k} &:= \mu_j(T^n_k\wedge t), & k=0,1,\ldots,\\
  \Delta\ln\mu_{j,k} &:= \ln\mu_{j,k} - \ln\mu_{j,k-1}, & k=1,2,\ldots,\\
  \pi_{j,k} &:= \pi(\mu_{j,k}), & k=0,1,\ldots,
\end{align*}
with the dependence on $n$ suppressed.
We can then regard
\begin{equation}\label{eq:nth}
  2\Gamma^{*,n}_{\pi}(t)
  :=
  \sum_{k=1}^{\infty}
  \sum_{j=1}^J
  \pi_{j,k-1}
  \Delta\mu_{j,k}^2
  -
  \sum_{k=1}^{\infty}
  \left(
    \sum_{j=1}^J
    \pi_{j,k-1}
    \Delta\ln\mu_{j,k}
  \right)^2
\end{equation}
as the $n$th approximation to $2\Gamma^*_{\pi}(t)$;
it can be shown that
\[
  2\Gamma^{*,n}_{\pi}(t)
  \to
  2\Gamma^*_{\pi}(t)
  \quad
  \text{ucqa}.
\]
Rewriting \eqref{eq:nth} as
\begin{equation*}
  2\Gamma^{*,n}_{\pi}(t)
  =
  \sum_{k=1}^{\infty}
  \sum_{j=1}^J
  \pi_{j,k-1}
  \left(
    \Delta\mu_{j,k}
    -
    \sum_{i=1}^J
    \pi_{i,k-1}
    \Delta\ln\mu_{i,k}
  \right)^2,
\end{equation*}
we can see that this expression is the cumulative variance of the logarithmic returns $\Delta\ln\mu_{j,k}$
over the time interval $[T^n_{k-1}\wedge t,T^n_k\wedge t]$
w.r.\ to the ``portfolio probability measure'' $Q(\{j\}):=\pi_{j,k-1}$.
This makes the expression~\eqref{eq:summary} very intuitive:
the excess growth rate of the portfolio $\pi$ over the naive expression
is determined by the volatility of the market weights w.r.\ to $\pi$.

As already mentioned,
the stochastic versions of Corollaries~\ref{cor:D-p} and~\ref{cor:D-p-simplified}
have been used for active portfolio management \cite{Fernholz:1999ETF}.
The remarks made above about the relation between Corollaries~\ref{cor:arbitrage} and~\ref{cor:arbitrage-simplified}
are also applicable to Corollaries~\ref{cor:D-p} and~\ref{cor:D-p-simplified};
the latter replaces the first addend on the right-hand side of \eqref{eq:D-p} by its lower bound.
Corollary~\ref{cor:D-p} decomposes the growth in the value of the diversity-weighted portfolio into two components,
one related to the growth in the diversity $\D_p(\mu)$ of the market weights
and the other related to the accumulation of the diversity-weighted variance of the market weights.
Corollary~\ref{cor:D-p-simplified} ignores the first component,
which does not make it vacuous since $\D_p$ is bounded,
always being between $1$ (corresponding to a market concentrated in one stock)
and $J^{(1-p)/p}$ (corresponding to a market with equal capitalizations of all $J$ stocks).

Several explanations have been suggested for the somewhat counterintuitive disjunction
stated at the beginning of this section:
\begin{itemize}
\item
  If we include all stocks traded in a real-world market in our model,
  perhaps making $J$ very large,
  the portfolio \eqref{eq:generated}
  and its special cases \eqref{eq:arbitrage-portfolio}, \eqref{eq:entropy-portfolio}, and \eqref{eq:D-p-portfolio}
  (particularly the last two) will not be efficient
  since they will be forced to invest into smaller and so less liquid stocks;
  it is known that portfolios generated by measures of diversity
  invest into smaller stocks more heavily than the market portfolio does
  \cite[Proposition~3.4.2]{Fernholz:2002}.
\item
  There is another explanation related to this common feature of the portfolios discussed in this paper that outperform the market
  (increased weights of smaller stocks as compared with the market).
  Over the last decades, such portfolios have been adversely affected by the tendency of larger companies
  to pay higher dividends
  (cf., e.g., \cite[Figure~7.4]{Fernholz:2002},
  describing the performance of an index that has been used in investment practice).
  The role of differential dividend rates in maintaining market diversity
  is emphasized in \cite{Fernholz:1999}.
\item
  If we restrict our attention only to $J$ largest stocks traded in a real-world market,
  for a moderately large $J$ (such as $J=500$ for S\&P 500),
  the performance of portfolios such as \eqref{eq:arbitrage-portfolio}, \eqref{eq:entropy-portfolio}, or \eqref{eq:D-p-portfolio}
  w.r.\ to this smaller ``market'' (which is now, in fact, a large cap market index)
  will be affected by the phenomenon of ``leakage''
  \cite[Example~4.3.5 and Figure~7.5]{Fernholz:2002}.
\end{itemize}

\section{Fernholz's master martingale and Stroock--Varadhan martingales}
\label{sec:additive}

One way to restate Theorem~\ref{thm:master} is to say that
\begin{equation}\label{eq:Fernholz}
  \frac{\S(\mu(t))}{\S(\mu(0))}
  \exp
  \left(
    -\frac12
    \sum_{i,j=1}^J
    \int_0^t
      \frac{D_{ij}\S(\mu)}{\S(\mu)}
    \dd[\mu_i,\mu_j]
  \right)
\end{equation}
is a value process q.a.;
in the terminology of \cite{\CTXIII}, it is a continuous martingale.
In this section we will see that Fernholz's master martingale~\eqref{eq:Fernholz}
is in fact a very natural object,
and not just a product of formal manipulations with the It\^o formula,
as might have appeared from its derivation in Section~\ref{sec:master}.
Connections with recent papers \cite{Karatzas/Ruf:2017} and \cite{Fernholz/etal:arXiv1608}
will be discussed later in the section.

Let $f$ be a $C^2$ function defined on an open neighbourhood $\dom f$ of $\Delta^J$ in $\R^J$.
The non-stochastic It\^o formula \cite{\CTXIII} implies that
\begin{equation}\label{eq:Stroock-Varadhan}
  f(\mu(t))
  -
  f(\mu(0))
  -
  \frac12
  \sum_{i,j=1}^J
  \int_0^t
    D_{ij} f(\mu)
  \dd[\mu_i,\mu_j]
  =
  \sum_{j=1}^J
  \int_0^t
    D_{j} f(\mu)
  \dd\mu_j
  \quad
  \text{q.a.},
\end{equation}
and so the left-hand side of \eqref{eq:Stroock-Varadhan} is a continuous martingale,
which we will refer to as the \emph{Stroock--Varadhan martingale} \cite[(5.4.2)]{Karatzas/Shreve:1991};
it is a non-stochastic version of the classical martingales used by Stroock and Varadhan
in their study of diffusion processes.

Fernholz's master martingale~\eqref{eq:Fernholz}
is the Dol\'eans exponential of the Stroock--Varadhan martingale on the left-hand side of~\eqref{eq:Stroock-Varadhan}
for $f:=\ln\S$.
Indeed, applying \eqref{eq:standard-1} gives the Dol\'eans exponential
\begin{multline*}
  \frac{\S(\mu(t))}{\S(\mu(0))}
  \exp
  \biggl(
    -
    \frac12
    \sum_{i,j=1}^J
    \int_0^t
      D_{ij} f(\mu)
    \dd[\mu_i,\mu_j]\\
    -
    \frac12
    \sum_{i,j=1}^J
    \int_0^t
      D_{i} f(\mu)
      D_{j} f(\mu)
    \dd[\mu_i,\mu_j]
  \biggr)
\end{multline*}
of the left-hand side of~\eqref{eq:Stroock-Varadhan},
which is equal, by the identity
\[
  D_{ij} f
  =
  \frac{D_{ij}\S}{\S}
  -
  \frac{D_{i}\S}{\S}
  \frac{D_{j}\S}{\S}
  =
  \frac{D_{ij}\S}{\S}
  -
  D_{i}f
  D_{j}f,
\]
to \eqref{eq:Fernholz}.
If we regard the Stroock--Varadhan martingale to be an additive process,
Fernholz's master martingale \eqref{eq:Fernholz} becomes its multiplicative counterpart.
``Additive'' and ``multiplicative'' is the terminology used in \cite{Karatzas/Ruf:2017,Fernholz/etal:arXiv1608}
(more carefully than in this paper),
and the relation between the Stroock--Varadhan martingale on the left-hand side of~\eqref{eq:Stroock-Varadhan}
and Fernholz's master martingale~\eqref{eq:Fernholz}
is somewhat analogous to the relation between the additive Bachelier formula \cite[Section~11.2]{Shafer/Vovk:2001}
and the multiplicative Black--Scholes formula \cite[Section~11.3]{Shafer/Vovk:2001} in option pricing.

The papers \cite{Karatzas/Ruf:2017} and \cite{Fernholz/etal:arXiv1608} study additive portfolio generation in depth.
In particular, these papers give numerous interesting examples
(including \eqref{eq:quadratic-2} below).

We can rewrite \eqref{eq:arbitrage-simplified} in Corollary~\ref{cor:arbitrage-simplified} as
\begin{equation*}
  Z_{\pi}(t)
  \ge
  \frac12
  \exp
  \left(
    \frac12
    \sum_{j=1}^J
    [\mu_j](t)
  \right)
  \quad
  \text{q.a.},
\end{equation*}
which implies
\begin{equation}\label{eq:arbitrage-A}
  Z_{\pi}(\tau_A)
  \ge
  \frac12
  e^{A/2}
  \quad
  \text{q.a.},
\end{equation}
where $A$ is a positive constant, $\tau_A$ is the stopping time
\begin{equation}\label{eq:tau}
  \tau_A
  :=
  \min
  \left\{
    t \st \sum_j [\mu_j](t) = A
  \right\},
\end{equation}
and $Z_{\pi}(\infty):=\infty$.

Qualitatively, \eqref{eq:arbitrage-A} means that the market satisfies Fernholz's arbitrage-type property:
we can beat a non-degenerate market (interpreting non-degeneracy as $\tau_A<\infty$ for all $A$).
The Stroock--Varadhan martingale on the left-hand side of~\eqref{eq:Stroock-Varadhan} also gives a Fernholz-type arbitrage,
which is, however, polynomial (and even linear) in $A$, unlike \eqref{eq:arbitrage-A}.
Setting
\begin{equation}\label{eq:quadratic-2}
  f(x)
  :=
  -
  \frac12
  \sum_{j=1}^J
  x_j^2
\end{equation}
(cf.\ \eqref{eq:quadratic-1}),
we can rewrite the continuous martingale on the left-hand side of~\eqref{eq:Stroock-Varadhan} as
\[
  Y(t)
  :=
  \frac12
  \sum_{j=1}^J
  \mu_j(0)^2
  -
  \frac12
  \sum_{j=1}^J
  \mu_j(t)^2
  +
  \frac12
  \sum_{j=1}^J
  [\mu_j]
  \ge
  -\frac12
  +
  \frac12
  \sum_{j=1}^J
  [\mu_j];
\]
therefore, $X:=2Y+1$ is a nonnegative continuous martingale satisfying $X(0)=1$ and
\begin{equation}\label{eq:SV}
  X(\tau_A)
  \ge
  A.
\end{equation}
Therefore, this is an alternative method for achieving the same qualitative goal,
$X(\tau_A)\to\infty$ as $A\to\infty$.
Quantitatively the result might appear weaker;
after all, we lose the exponential growth rate in $A$.
However, there is a range of $A$ (roughly between $0.7$ and $4.3$)
where the Stroock--Varadhan martingale $X$ performs better:
see Figure~\ref{fig:comparison}.

\begin{figure}[tb]
  \begin{center}
    \includegraphics[width=0.6\textwidth]{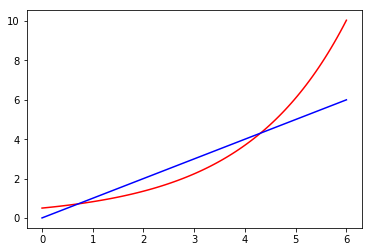}
  \end{center}
  \caption{The values of the Fernholz martingale $Z_\pi$ (red curve)
    and the Stroock--Varadhan martingale $X$ (blue straight line) at time $\tau_A$;
    the horizontal axis is labelled by the values of $A$}
  \label{fig:comparison}
\end{figure}

\section{Conclusion}
\label{sec:conclusion}

Figure~\ref{fig:comparison} gives two functions $g$ such that a final capital of $g(A)$ is achievable at time $\tau_A$.
It would be interesting to characterize the class of such functions $g$.
A related question is:
what is the best growth rate of $g(A)$ as $A\to\infty$?
This question can be asked in both stochastic and non-stochastic settings.
These are some directions of further research for non-stochastic theory:
\begin{itemize}
\item
  A natural direction is to try and strip other results of stochastic portfolio theory
  of their stochastic assumptions.
  First of all, it should be possible to extend Theorem~\ref{thm:master} to functions $\S$
  that are not smooth (as in \cite[Theorem~4.2.1]{Fernholz:2002});
  the existence of local time in a non-stochastic setting is shown in \cite{Perkowski/Promel:2015local}
  and, in the case of continuous price paths, can be deduced from the main result of \cite{\CTIV}.
\item
  Another direction is to extend this paper's results to general num\'eraires
  (this paper uses the value of the market portfolio as our num\'eraire).
\item
  Finally, it would be very interesting to extend some of the results to c\`adl\`ag price paths.
\end{itemize}

\subsection*{Acknowledgements}

The first draft of this paper was prompted by discussions with Martin Schweizer and D\'aniel B\'alint in December 2017.
I am grateful to Ioannis Karatzas for information on the existing literature
(in particular, he drew my attention to connections between Section~\ref{sec:additive} of this paper
and additive portfolio generation in \cite{Karatzas/Ruf:2017} and \cite{Fernholz/etal:arXiv1608})
and for his advice on presentation, notation, and terminology.
Thanks to Peter Carr, Marcel Nutz, Philip Protter, and Glenn Shafer for useful comments.

\appendix
\section{Connections with the foundations of game-theoretic probability}

In this appendix we will see yet another method of achieving the qualitative goal
of $\lim_{A\to\infty}X(\tau_A)=\infty$ for a nonnegative supermartingale $X$.
The result will be weaker than both functions in Figure~\ref{fig:comparison},
but it will shed light on a seemingly paradoxical feature of continuous-time game-theoretic probability.

The method uses the non-stochastic Dubins--Schwarz theorem presented in \cite{\CTIV}
and is based on the following apparent paradox,
which we first discuss informally.
As agreed in Section~\ref{sec:market},
we regard $\mu_1,\ldots,\mu_J,1$ as tradable securities.
According to the non-stochastic Dubins--Schwarz theorem and a standard property of Brownian motion,
with very high lower probability all $J$ securities will eventually hit zero
if their volatility is appreciable.
When this happens, the normalized value of the market $\mu_1+\cdots+\mu_J$ will be 0 rather than 1,
which is impossible.
Therefore, we expect an event of a low upper game-theoretic probability to happen,
i.e., we expect to be able to outperform the market.
This is formalized in the following statement:

\begin{proposition}
  For any constant $A>0$, there is a nonnegative supermartingale $X$ such that $X(0)=1$ and
  \begin{equation}\label{eq:arbitrage-our}
    X(\tau_A)
    \ge
    1.25
    J^{-3/2} A^{1/2}
    \quad
    \text{q.a.},
  \end{equation}
  where $\tau_A$ is the stopping time \eqref{eq:tau}
  and $X(\infty)$ is interpreted as $\infty$.
\end{proposition}

\begin{proof}
  For each $j\in\{1,\ldots,J\}$,
  we will construct a nonnegative continuous martingale $X_j$ satisfying $X_j(0)=1$ and
  \begin{equation}\label{eq:arbitrage-our-j}
    X_j(\tau_j)
    \ge
    1.25
    (A/J)^{1/2}
    \quad
    \text{q.a.},
  \end{equation}
  where
  \begin{equation*}
    \tau_j
    :=
    \min
    \{
      t \st [\mu_j](t) = A/J
    \}.
  \end{equation*}
  (In this case we can set $X$ to the average of all $J$ of $X_j$ stopped at time $\tau_j$.)
  According to \cite[(2.6.2)]{Karatzas/Shreve:1991},
  the probability that a Brownian motion started from 1
  (in fact $\mu_j$ is started from $\mu_j(0)\le 1$)
  does not hit zero over the time period $A/J$ is
  \[
    1 - \sqrt{\frac{2}{\pi}} \int_{(J/A)^{1/2}}^{\infty} e^{-x^2/2} \dd x
    =
    \sqrt{\frac{2}{\pi}} \int_0^{(J/A)^{1/2}} e^{-x^2/2} \dd x
    \le
    \sqrt{\frac{2}{\pi}} (J/A)^{1/2}.
  \]
  In combination with the non-stochastic Dubins--Schwarz result \cite[Theorem~3.1]{\CTIV}
  applied to $\mu_j$,
  this gives \eqref{eq:arbitrage-our-j} with
  \[
    \sqrt{\frac{\pi}{2}}
    >
    1.25
  \]
  in place of $1.25$.
\end{proof}

The processes $Z_{\pi}$ in \eqref{eq:arbitrage-A} and $X$ in \eqref{eq:SV} are nonnegative supermartingales
in the sense of \cite{\CTXIII}
(in fact, nonnegative continuous martingales).
On the other hand, the process $X$ in \eqref{eq:arbitrage-our}
is a nonnegative supermartingale in the sense of the more cautious definitions in \cite{\CTIV}.
This can be regarded as advantage of \eqref{eq:arbitrage-our} over \eqref{eq:arbitrage-A} and \eqref{eq:SV}.
However, a disadvantage of \eqref{eq:arbitrage-our}
is that quantitatively it is much weaker than both \eqref{eq:arbitrage-A} and \eqref{eq:SV};
the right-hand side of~\eqref{eq:arbitrage-our} is always smaller than the right-hand side of~\eqref{eq:arbitrage-A},
and it is greater than the right-hand side of~\eqref{eq:SV} only for a small range of $A$
(approximately $A\in(0,0.2)$ for $J=2$).
\end{document}